\documentclass[lnicst]{svmultln}
\usepackage{makeidx}  
\usepackage{amsmath,amssymb} 
\usepackage{color}           
\usepackage{graphicx}
\usepackage{caption}
\usepackage{subcaption}  
\usepackage{epstopdf}        
\usepackage{url}
\usepackage[ruled,linesnumbered]{algorithm2e}
\begin{document}
\mainmatter              
\title{A Game-Theoretic Model for Analysis and Design of Self-Organization Mechanisms in IoT}
\titlerunning{Game-Theoretic Model of Self-Organization in IoT}  
%
\author{Vahid Behzadan \and Banafsheh Rekabdar}
\authorrunning{Vahid Behzadan et al.}   
%
\tocauthor{Vahid Behzadan\inst{1}, Banafsheh Rekabdar}
\institute{Department of Computer Science and Engineering\\ University of Nevada, Reno\\
	1664 N Virginia St, Reno, NV 89557\\
\email{vbehzadan@unr.edu , brekabdar@unr.edu}}

\maketitle              

\begin{abstract}        
We propose a framework based on Network Formation Game for self-organization in the Internet of Things (IoT), in which heterogeneous and multi-interface nodes are modeled as self-interested agents who individually decide on establishment and severance of links to other agents. Through analysis of the static game, we formally confirm the emergence of realistic topologies from our model, and analytically establish the criteria that lead to stable multi-hop network structures.
\end{abstract}
\section{Introduction and Motivation}\label{intro}
Through the past decade, the number of internet-enabled devices has been growing at an unprecedented rate. The paradigm of Internet of Things (IoT) envisions an even more immersive and pervasive exploitation of internet connectivity by enabling more objects and devices to connect. Emerging applications of this move towards ubiquitous connectivity are wide and vast \cite{bandyopadhyay2011internet}, ranging from domestic monitoring and smart home solutions to healthcare solutions \cite{rohokale2011cooperative}, smart grids\cite{yun2010research}, and disaster monitoring \cite{fang2014integrated}. It hence follows that instances of IoT will be comprised of a great number of various devices, each with unique requirements and capabilities, leading to heterogeneity both in terms of function and communications.

The inevitably high degree of heterogeneity and scalability of IoT, dim the odds of feasibility and scalability for centralized control approaches \cite{ma2015networking}. An alternative to centralized architectures for IoT are those that rely on autonomic management of connectivity and resources through self-configuration \cite{athreya2013network}. Such solutions model the network as a system comprised of individual agents, each of which aims to retain connectivity with the network while optimizing their objectives, such as energy consumption and throughput. Even though this multi-agent abstraction presents a promising approach towards scalability, the decentralized nature of self-configuring IoT gives rise to many critical challenges in mechanism design. Of the most critical of these challenges is the problem of topology control, which is further complicated by the heterogeneity of IoT devices. Owing to the similarity of distributed IoT and Ad Hoc networks, the literature on self-organization and topological analysis of IoT are mainly focused on adopting techniques that are originally developed for generic distributed networks such as Wireless Sensor Networks (WSNs) \cite{athreya2013network}. Yet, unique features, such as the immense diversity in capabilities and requirements in all aspects of IoT present major distinguishing factors that necessitate the development of techniques specific to the challenges of this emerging technology. 

\begin{figure}[!t]
	\centering
	\includegraphics[width=\columnwidth]{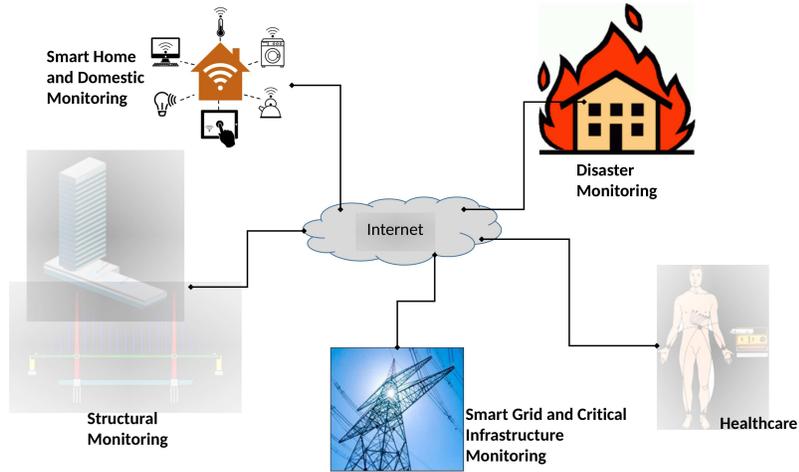}
	\caption{Applications of IoT}
	\label{fig_apps}
\end{figure}

The multi-agent model of IoT is comprised of opportunistic devices that aim to maximize their success in fulfilling their individual objectives, such as preservation of connectivity to the network, minimization of energy consumption and maintenance of a minimum Quality of Service (QoS). The inherent limitation of resources available to such opportunistic agents in any real-world deployment of IoT gives rise to a competitive environment, which motivates a game theoretic investigation of interactions in self-organizing IoT. The application of game theory to distributed topology control and self-organization has been an active area of research in recent years. Some of the notable literature in this area include the work of Eidenbenz, et. al. \cite{eidenbenz2006equilibria} on the analysis of equilibria in topology control games, Nahir, et. al.'s  detailed investigation of applying game theory to various problems of topology control \cite{nahir2009topology}, and Saad et. al. proposal's of a game theoretic algorithm for cooperative relaying in \cite{saad2011network}, based on their earlier analysis of the formation of hierarchical topologies in  multi-hop networks \cite{saad2009hierarchical}. The models presented in these and many other topology control games, one critical limitation is the assumption on homogeneity of the network. Recently, Meirom et. al. proposed a model of topology control games for heterogeneous AS-Level networks \cite{meirom2016strategic} \cite{meirom2015formation}, which  considers some degree of heterogeneity, but only accounts homogeneous link costs.

Based on the inevitable emphasis on the connectivity aspects of IoT networks,  this paper builds on the aforementioned models to provide a framework for analysis and design of distributed topology control mechanisms in IoT. The proposed framework is based on modeling of self-organization as a Network Formation Game \cite{jackson2005survey}, in which the actions of players are establishment or severance of links with other nodes in the IoT. Contrary to previous models, we consider heterogeneity in both communications and link cost. The proposed model also accounts for nodes equipped with multiple communication interfaces, thus supporting modern devices such as smart phones. We provide an analytical derivation of the criteria required for formation of a clique topology between nodes that are directly connected to the internet, and further develop this analysis to present the necessary criteria which lead to formation of hierarchical and star topologies between internet-connected nodes and the rest of the network. 

The remainder of this paper is organized as follows: Section \ref{model} details the model of IoT networks, followed by the formulation of network formation game in Section \ref{Game}. Emergence of stable IoT topologies and their criteria is discussed in Section \ref{analysis}, and the results of a numerical simulation is presented in \ref{system}. Finally, Section \ref{conclusion} concludes the paper with remarks on future areas of work.

\section{IoT Network Model}\label{model}
The generic definition of IoT has given rise to numerous models for the network structure and architecture \cite{ma2015networking}. In this work, IoT is considered to be a network formed with the objective of enabling direct or relayed connectivity of heterogeneous nodes to the internet (or other backbone networks). Heterogeneity of nodes entails diverse hardware and software parameters throughout the network, such as the number and type of communication interfaces (e.g. WLAN, LTE, Ethernet, etc.), energy constraints, and bandwidth requirements.

Accordingly, we model the IoT as a network $G(P)$ of $N$ nodes $P=\{P_i\ | \forall i \in \{1,2,...,N\}\}$, each with an arbitrary number of single channel radio interfaces. This definition may be seamlessly extended to cover multi-channel radios as well, via representing each as a group of single-channel radios. It is assumed that all interfaces of a node can be active simultaneously, but as detailed in Section \ref{Game}, the effects of activating each additional interface on undesired aspects such as co- and cross-interference, channel congestion, and energy consumption may be suitably captured in the system cost function. The presented model also allows that some, or all of the interfaces in nodes may remain idle throughout the analyzed operation. 

As the focus of this study is on topological properties, it is assumed that nodes are static relative to each other. Also, we consider the case that every node in the network is aware of its distance in terms of number of intermediate hops with every other node in the network. This can be justified by reliance on routing tables obtained from proactive network layer protocols such as OLSR \cite{clausen2003optimized}. The extent of a node's knowledge of the overall network topology is assumed to be limited to its directly connected neighbors.  

Nodes are classified in two categories: Those with direct connectivity to the internet, such as WiFi Access Points and 3G/LTE Enabled Devices, and those which need to be connected to the internet via the nodes in the former group, such as Bluetooth/Zigbee sensors. Let the set of Internet Connected (IC) nodes $G_I \in P$ denote the set of nodes with direct connection to the internet, and the set of non-ICs $G_S \in P \backslash G_I$ is the set of nodes that do not have a direct connection to the internet. The emerging network is thus hierarchical with at least two tiers: a higher tier formed by IC nodes, and a lower tier comprised of non-IC nodes who aim to connect to the higher tier. Hence, an important objective of IoT network controllers, whether centralized or distributed, is to enable the connection to the internet to the non-IC node, via linking them to one or more IC nodes. In line with practical network protocols, a further limit is imposed to the maximum number of hops that may exist between each pair of nodes, denoted by $h_{Max}$. The following section provides the details of one such controller based on a game theoretical framework known as Network Formation Games.

\section{Game Formulation} \label{Game}
Formation of macro-scale topologies in distributed networks is the collective result of the individual decisions made by each nodes on which set of nodes to connect with, and which links to severe. With the assumption that every such node aims at gaining more utility from its decisions and consequent actions, this interaction of multiple decision makers can be formulated as a Network Formation Game \cite{jackson2005survey}. Such games are comprised of competing agents who control the set of nodes they are connected to, with the common objective of forming coalitions of nodes that is most profitable for the deciding agent. It is evident that the game being considered in this work is of the non-cooperative type, since the decisions are made independently. Another assumption adopted in our proposal is that a link between two nodes is established if, and only if, both nodes consent to its establishment. This assumption emulates the real-world phenomenon that occurs in cost-optimizing distributed networks. A simple, yet realistic example is depicted in figure (\ref{fig_graph}. This figure illustrates a network formation game in which the objective of all players is to minimize their cost while maintaining their reachability from any other player by at most one intermediate hop - a property that we shall label as one-hop-reachable. The cost incurred to each player of this game is the cost of establishing their immediate links (denoted by edge weights in figure (\ref{fig_graph}), which is assumed to be the same for both of the linked nodes. If two nodes are not one-hop-reachable, their cost is set to be infinity. For instance, the cost incurred to node $B$ is the cost of establishing the link $BC$ plus the cost of establishing $BD$, i.e. $2+4=6$. As is shown in the figure, for node $C$ to be one-hop-reachable to node $A$, the minimum cost is obtained by relaying through node $D$. Yet for node $D$, establishment of a link to node $C$ does not bring any utility but losses, as node $D$ has already established a cheaper path to $C$ via node $B$, and is directly connected to node $A$. Hence, node $D$ will not consent to spending its limited bandwidth and energy to relay a transmission that gains him no benefits. Consequently, nodes $A$ and $C$ settle on establishing an expensive direct link to avoid the infinite cost of unreachability.

\begin{figure}[t]
	\centering
	\includegraphics[width=\columnwidth]{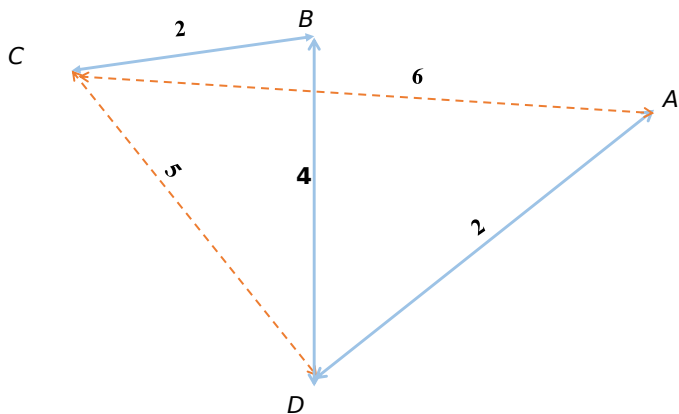}
	\caption{Example of the mutual consent in Myerson games}
	\label{fig_graph}
\end{figure}

Network formation games that are based on consensual establishment of links are known as bilateral linking or Myerson games \cite{dutta2013networks}, which is a widely adopted model in game theoretic distributed topology control, mainly due to its agreement with the opportunistic behavior of agents in decentralized networks. Our proposed framework builds atop of the previous work on bilateral link formation by extending the application of the Myerson model to considerations beyond that of minimizing energy consumption as the sole objective of the game, replacing the abstracted link establishment parameters with those of real wireless interface characteristics and propagation model, and filling the gap in self-organizing IoTs by providing a novel cross-layer framework for analytical design and evaluation of protocols and parameters involved in the distributed formation of IoT topologies. 

Even though the real phenomenon of network formation in ad hoc communications networks is of a dynamic nature, this work concentrates on the analysis of a static bilateral linking game, with the aim of gaining insights on the characteristics of emerging stable topologies, along with the criteria that leads to their emergence. Similar to every other game, our proposed Myerson game is formed of players, set of strategies, and a payoff/cost structure, the details of each are presented in this section.

\subsection{Players}
Let $P = \{p_1,p_2,...p_N\}$ denote a group of $N$ agents. Each agent $p_i$ is characterized by the following features:

\begin{itemize}
	\item Ordered set of its radio interfaces $R_i$, where $|R_i|$ is the number of interfaces and $R_i(r\in \{1,2,...,|R_i|\}) \in \{0,1\}$ is a binary value, indicating whether the interface is currently being used or not.
	\item Frequency of operation for each radio interface $f_{i,r}$
	\item Maximum bandwidth for each radio interface $b_{i,r}$
	\item Minimum required bandwidth $b_i ^ {Max}$ 
	\item Maximum transmit power for each radio interface $\tau _{i,r}$
	\item Receiver sensitivity for each radio interface $\S_{i,r}$
	\item Maximum antenna gain for each radio interface $x_{i,r}$
	\item 2-D Position $\gamma _i = (x_i, y_i)$
	\item Feature tuple for each interface $w_{i,r} = (f_{i,r}, b_{i,r}, b_i ^ {Max}, \tau _{i,r}, S_{i,r}, x_{i,r}, \gamma _i)$
\end{itemize}

Define the network topology $G = \{g_{ij}: i,k \in P , i \neq j \}$. If a bidirectional link is established between $p_i$ and $p_j$, then $g_{ij} = (r_i,r_j)$, where $r_i \in R_i$ is the interface chosen by the node $i$ to communicate with the corresponding interface in node $j$, i.e. $r_j \in R_j$. If there is no direct link between $i$ and $j$, $g_{ij} = (-1,-1)$.

\subsection{Strategies}
Let $C_i$ denote the cost function for every node $p_i$. Any node $p_i \in P$ may form a link $g_{ij}=(r_i,r_j)$ with any Node $p_j$ in the neighborhood $M(i)$, defined as the set of all nodes that fall within the maximum communications range of $i$, if:

\begin{enumerate}
	
	\item Nodes must have at least one type of radio interface available and in common, i.e. :
	
	\item $\Delta C (p_i, G + (r_i,r_j)) < 0$
	\item $\Delta C (p_j, G + (r_i,r_j)) < 0$

\end{enumerate}

Where $	\Delta C(p_K, G + (r_k, r_l)) = C (p_k, G \cup \{(r_k,r_l)\}) - C (p_k, G)$ is the difference between the total cost to node $p_k$ by establishing the link $(r_k, r_l)$ and the total cost to $p_k$ without the establishment of this link.

Agent $p_i$ may remove a link with agent $p_j$ in $M' \subset R_k ^c \{R_i (k) \neq (-1,-1)\}$ if:
\begin{eqnarray}
\Delta C (p_i, G - (r_i,r_j)) < 0 \nonumber
\end{eqnarray}

Where $\Delta C (p_k, G - (r_k,r_l)) = C (p_k, G \backslash \{(r_k,r_l)\}) - C (p_k, G)$ is the difference between the cost incurred by node $p_k$ removing the link $(r_k,r_l)$ and the cost incurred by maintaining this link.

\subsection{Payoff Structure}
For each node $p_i$, the payoff of forming a direct link is dependent on the set of objectives listed below:

\begin{enumerate}
	\item Minimize the total cost of link establishment $\sum_{z=1}^{{deg(p_i)}} L'_i(z)$
	\item Minimize the hop distance to all nodes in the network, with priority over minimizing distance to the nodes directly connected to the internet.
	\item Minimize energy consumption by avoiding excessive relay transmissions 
\end{enumerate}

The corresponding cost function for each node is thus formulated as:

\begin{eqnarray}
C(i,G) = C_i = \sum_{z=1}^{{deg(p_i)}} L'_i(z) + \Gamma \sum_{j\in G_I} h(i,j) \nonumber \\
+ \sum_{k\in G_S} h(i,k) + B_i 	
\end{eqnarray}

Where $L'_i(z)$ is the cost of establishing the z-th link of $p_i$, with $deg(p_i)$ denoting the number of links established by $p_i$. Let $L_i(z)$ be the link between nodes $i$ and $z$. To model the link cost, the following factors are considered:

\begin{itemize}
	\item $L_i (z)$ is directly proportional to the minimum transmission power required for $z$ to receive the signal. The transmission power depends on the fading model and noise on the channel, which generally is inversely proportional to the Euclidean distance between nodes, their antenna gains, and the receiver sensitivity. Every interface has a maximum budgeted transmit power , beyond which $L_i(z) = \infty$
	\item $L-i (z)$ is directly proportional to the number of connections established on interface $R_i (r)$. The more this number is, the more congestion is expected and hence the throughput suffers.
	\item $L_i (z)$ is inversely proportional to bandwidth. Higher the bandwidth, higher the throughput will be.
\end{itemize}

Hence, a generic formulation for $L_i (z)$ is constructed as:
\begin{eqnarray}	
\sum_{z=1}^{{deg(p_i)}} L'_i (z) = \sum_{r=1}^{|R_i|}  R_i (r) \sum_{z\in P s.t. g_{iz} = (R_i(r),o)} \alpha . \rho_i . \frac {\sigma_{ir}}{\beta_{ir}} 
\end{eqnarray}

Where $\alpha$ is a constant factoring the effect of each additional link on interface $R_i(r)$, $\rho_i$ is the relative importance of preserving energy to achieving the desired throughput, $\sigma_{ir}$ is the power transmitted by $p_i$ on this link, and $\beta_{ir}$ is the ratio of the available bandwidth to the required bandwidth, i.e.:
\begin{equation}
\beta_{ir} = \frac{b_{ir}}{b_i^{Min}}
\end{equation}

The factor $\Gamma \geq 1$ is the weighting factor for tuning the emphasize on minimizing the shortest hop-distance $h(i,j)$ to every IC node $j\in G_I$. $B_i$ is the bridging coefficient of node $p_i$, estimating the local burden of bridging communities and thus modeling the relative amount of relay transmissions that $p_i$ may have to handle for its neighbors. It is shown in \cite{hwang2008bridging} and \cite{komali2006distributed} that the higher values of bridging coefficient represent a higher risk of congestion, as well as collisions. Bridging coefficient is calculated as:
\begin{equation}
B_i = \frac{\frac{1}{deg(i}}{\sum_{j\in \{k : g_{ik} \neq (-1,-1)\}} \frac{1}{deg(j)}}
\end{equation}

\section{Equilibrium Topologies in Static game}\label{analysis}
This section investigates the criteria which enable the emergence of stable and efficient topologies from the proposed network formation mechanism. Having a game theoretic abstraction of the problem, we study the characteristics of stable networks by analyzing the equilibria of our model. One of the most intuitive types of equilibrium is the Nash equilibrium, defined as strategy profiles at which no player can increase its profit by unilaterally deviating from that profile, hinting at a stable outcome. Yet, Nash equilibrium is shown to be a weak notion for stability in network formation games \cite{bloch2006definitions}. Considering the bilateral nature of link formation in such games, stability of outcomes is characterized more accurately by considering bilateral deviations. To satisfy this requirement, we consider the notion of \emph{pairwise stability} \cite{bloch2006definitions}. A strategy profile is said to be pairwise stable if no unilateral or bilateral deviations could increase the utility of the players. Formally, a topology $G$ is pairwise stable if the following conditions are met:
\begin{enumerate}
	\item $\forall i, ij\in G, C(i,G) \leq C(i,G-ij)$
	\item $\forall i,j\notin G, if C(i,G+ij) < C(i,G) then C(j,G+ij) > C(j,G)$
\end{enumerate}

In the following subsections, we utilize pairwise stability in the formal analysis of stable topologies that can emerge from the proposed model. 

\subsection{Formation of Cliques}
A notable number of recent literature on bilateral link formation games are based on models that result in systematical limitation of pairwise stability to forest and tree topologies (e.g.\cite{arcaute2007dynamics}, \cite{arcaute2009network}) This property greatly neuters the applicability of such models to IoT. As discussed in Section \ref{model}, nodes in IoT are categorized as either Internet-Connected (IC) or non-IC. It is intuitive to assume that each IC node is directly connected to every other IC nodes through the internet connection, thereby the set of all IC nodes inherently forms a clique. Therefore, if the cost of link establishment is bounded by a critical value, it is expected that the clique remains stable. In the following theorem, we prove that under certain criteria, this topology is indeed pairwise stable.

\newtheorem{thm}{Theorem}
\begin{thm}
	Let $L_i(k)$ be the maximum cost for any internet-connected node $p_i\in G_I$ to establish a link with node $p_k\in G_I$. If $L_i(k) < \Gamma -1$, then the nodes in $G_I$ form a clique.
\end{thm}
\begin{proof}
	Assume a node $p_i$ that is yet to establish connections to any node in $G$. For any node $p_k\in G_I$, the cost difference of establishing a link is given by:
	
	\begin{eqnarray}
	C(p_i, G+g_{ik}) = C(p_i, G\cup \{r_i,r_k)\})-C(p_k,G) \nonumber \\
	= L_i(k)+\Gamma (-1)+0+\Delta B_i
	\end{eqnarray}

	Where
	\begin{eqnarray}
	\Delta B_i = \frac{\frac{1}{deg(i)+1}}{\sum_{j\in \{\forall \zeta |g_{i\zeta} \neq (-1,1)\}}\frac{1}{deg(j)}+1}\nonumber \\
	- \frac{\frac{1}{deg(i)}}{\sum_{j\in \{\forall \zeta |g_{i\zeta} \neq (-1,1)\}}\frac{1}{deg(j)}}
	\end{eqnarray}

	Considering the minimum and maximum values of $deg(i)$ and $\sum_{j\in \{\forall \zeta |g_{i\zeta} \neq (-1,1)\}}\frac{1}{deg(j)}\}$, it is trivial to show that:
	
	\begin{equation}
	0<\Delta B_i<1 \nonumber
	\end{equation}
	
	Hence, the maximum valid value of the cost difference is given by:
	\begin{equation}
	\Delta C(p_i, G+g_{ik}) = L_i(k) - \Gamma + 1 
	\end{equation}

	For this cost difference to be feasible for all nodes in $G_I$, the following condition must be satisfied:
	
	\begin{eqnarray}
	\Delta C(p_i, G+g_{ik}) < 0 \nonumber\\
	\Rightarrow L_i(k) - \Gamma + 1 < 0 \nonumber \\
	\Rightarrow L_i(k) < \Gamma - 1
	\end{eqnarray}
	
	If this condition holds true, establishment of a link between any pair of nodes in $G_I$ decreases the cost for both nodes, hence leading to a clique topology. Inversely, severing any link in the resulting clique by any node $i\in G_I$ would impose a higher cost to $i$ than gain. Therefore, this criteria leads to cliques that are pairwise-stable.
\end{proof}

\subsection{Formation of Stars and Hierarchies}
Having established the criteria for the proposed model to result in a realistic stable topology for IC nodes, we study the topologies that emerge under this criteria for non-IC nodes. First, we derive the conditions that result in every non-IC node being linked to at most one of the IC nodes. Then, we derive the necessary conditions for formation of star clusters between non-IC nodes and IC nodes.

\newtheorem{thm2}{Theorem}
\begin{thm} If $L_i(k) < \Gamma - 1$, the maximum number of links between any non-IC node $j\in G_S$ and the set of Internet-connected nodes $G_I$ is 1.
\end{thm}
\begin{proof}
	Assuming there already exists a link between $i\in G_I$ and $j\in G_S$, the maximum cost difference of establishing a second link from another node $i'\in G_I\backslash \{i\}$ to $j$ is:
	\begin{eqnarray}
	\Delta C(i', G+g_{i'j}) = L_{i'}(j) - \Gamma + 0 + 1
	\end{eqnarray}
	
	For this link to be feasible for $i'$, the following condition must be met:
	\begin{eqnarray}
	\Delta C(i', G+g_{i'j}) < 0 \nonumber \\
	\Rightarrow L_{i'}(j) < \Gamma - 1
	\end{eqnarray}
	
	Therefore, if the minimum cost of connection to a node $j\in G_S$ satisfies $L_{i'}(j) > \Gamma - 1$, every non-IC node is connected to at most one IC node.
\end{proof}

In the following theorem, we derive the conditions under which every non-IC node is directly connected to an IC node, thus forming star-shaped clusters whose centers are IC nodes.

\newtheorem{thm3}{Theorem}
\begin{thm}
	Let $L_i(k) < \Gamma - 1$ and $L_{i'}(j) > \Gamma - 1$, the maximum degree of any non-IC node $j\in G_S$ is 1 iff $\forall j'\in G_S\setminus{j}, L_j(j') > \frac{1}{2}$.
\end{thm}
\begin{proof} Theorem 2 proves that under the aforementioned conditions, the maximum number of links between any non-IC node and all IC nodes is 1. Assume that $j$ establishes is a second link to a node $j'\in G_S$. The cost difference is given by:
	\begin{eqnarray}
	\Delta C(j,G+g_{jj'}) = L_j(j') + 0 - 1 + \frac{1}{2} \nonumber\\
	\end{eqnarray}
	
	For this action to be infeasible, the cost difference must be positive. Therefore:
	
	\begin{eqnarray}
	L_j(j') + 0 - 1 + \frac{1}{2} > 0 \nonumber\\
	\Rightarrow L_j(j') > \frac{1}{2}
	\end{eqnarray}
\end{proof}

As a corollary of Theorem 3, it is worth noting that if $G$ is connected and the conditions of theorems 1 and 2 are satisfied, but condition of theorem 3 is not, then the resulting topology contains nodes that have one link to the IC set, but are connected to one or more non-IC nodes. Such nodes act as gateways and relays for other non-IC nodes connected to them, and the emerging topologies have more than the original 2 levels of hierarchy, namely IC and non-IC. Consequently, this model allows for resource planning by determination of nodes that are bound to become relays, and therefore require higher communications and processing capabilities.

\section{Numerical Results}\label{system}
To demonstrate the performance of our proposed model in IoT, a smart home network of 10 nodes is simulated. The network is comprised of 4 IC nodes connected to the internet through WLAN 802.11n, and equipped with Bluetooth and Z-Wave \cite{gomez2010wireless} interfaces, 3 non-IC nodes with Bluetooth and Z-Wave interfaces, and 3 non-IC nodes with only Z-Wave interfaces. Nodes are positioned on a $30m\times 30m$ grid, as depicted in figure \ref{fig_Sim1}. Parameters of each radio interface are presented in Table \ref{table:sim_par}. 
\begin{table}
	\caption{Simulation Parameters}
	\label{table:sim_par}
	\centering
	\begin{tabular}{lll}\hline
		Parameter                   & Symbol      & Value      \\ \hline
		Freq. of WLAN            & $f_{i,0}$            & $2.4GHz$           \\
		Freq. of Bluetooth       & $f_{i,1}$            & $2.4GHz$           \\
		Freq. of Z-Wave          & $f_{i,2}$            & $0.908GHz$           \\
		Max. Bitrate of WLAN     & $b_{i,0}$            & $300Mbps$           \\
		Max. Bitrate of Bluetooth   & $b_{i,1}$            & $2 Mbps$           \\
		Max. Bitrate of Z-Wave      & $b_{i,2}$            & $40Kbps$           \\
		Min Req. Bitrate of IC Nodes & $b_i^{Min}$            & $10Mbps$           \\
		Min Req. Bitrate of Bluetooth Nodes & $b_i^{Min}$            & $0.5Mbps$           \\
		Min Req. Bitrate of Z-Wave Nodes & $b_i^{Min}$            & $5Kbps$           \\	
		Max. Tx Power of WLAN    & $\tau_{i,0}$            & $1W$           \\	
		Max. Tx Power of Bluetooth    & $\tau_{i,1}$            & $25mW$           \\		
		Max. Tx Power of Z-Wave    & $\tau_{i,2}$            & $1mW$           \\		
		Rx Sensitivity of WLAN    & $\S_{i,0}$            & $10^{-11}W$           \\		
		Rx Sensitivity of Bluetooth    & $\S_{i,1}$            & $10^{-10}W$           \\			
		Rx Sensitivity of Z-Wave    & $\S_{i,0}$            & $6.3\times 10^{-13}W$           \\			
		Max. Hop Distance    & $h_{Max}$            & $5$           \\
		\hline \\
	\end{tabular}
\end{table}

\begin{figure}[!h]
	\centering
	\begin{subfigure}[b]{0.55\textwidth}
		\includegraphics[scale = 0.40]{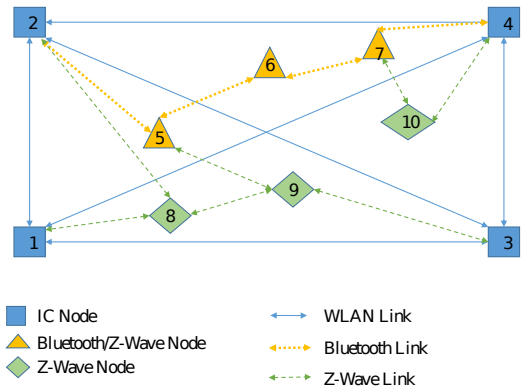}
		\caption{}
		\label{fig_Sim1}
	\end{subfigure}
	\begin{subfigure}[b]{0.55\textwidth}
		\includegraphics[scale = 0.40]{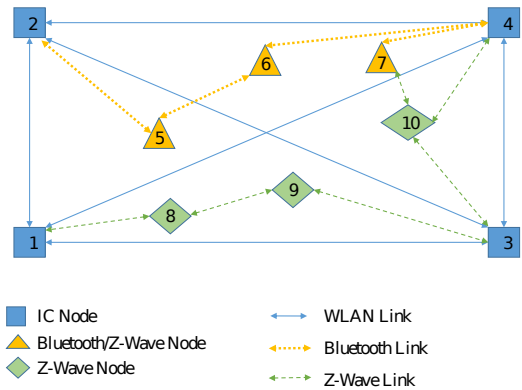}
		\caption{}
		\label{fig_Sim2}
	\end{subfigure}
	\caption{(a)Formation of Clique between IC Nodes, (b)Formation of Hierarchy between Non-IC Nodes}
\end{figure}
Considering the highest link cost, the condition for formation of Cliques presented in Theorem 1 is satisfied by setting  $\Gamma = 570$. As depicted in Figure X, the resulting topology is a clique for the IC nodes, while non-IC nodes represent a hierarchical  structure. Also, it is observed that the non-IC nodes $7$ and $9$ are connected to more than one IC node. By increasing the value of $\Gamma$ to $600$, the model results in a topology that satisfies the conditions of both theorems 1 and 2, where nodes $6$ and $7$ act as relays for nodes $8$ and $9$, respectively.

\section{Conclusions}\label{conclusion}

In this paper, we proposed a model for self-organization in IoT based on bilateral link formation strategies. The model captures the heterogeneity of devices in IoT, as well as the emphasis on connectivity to the internet in the proposed cost function. The subsequent analysis of the static game established the criteria for emergence of cliques between the set of internet-connected nodes, as well as multi-hop and star structures. 
Following the proposed model, further analysis of the static game may provide insights into the efficiency of emerging topologies, and establish the criteria for derivation of optimal network structures. Furthermore, this model provides a foundation for design and evaluation of dynamic games and algorithms for distributed self-organization in heterogeneous networks such as IoT.  
\end{document}